\documentclass[a4paper]{article}

\usepackage{amsmath,amsfonts,amssymb,amsthm}
\usepackage{fullpage}
\usepackage[all]{xy}
\SelectTips{cm}{}

\newtheorem{theorem}{Theorem}[section]
\newtheorem{proposition}[theorem]{Proposition}

\theoremstyle{definition}
\newtheorem{definition}[theorem]{Definition}
\newtheorem{remark}[theorem]{Remark}
\newtheorem{example}[theorem]{Example}
\newtheorem{assumption}[theorem]{Assumption}

\newcommand{\rto}{\xymatrix@C=1pc{ \ar[r] & } } 
\newcommand{\lto}{\xymatrix@C=1pc{ & \ar[l] } } 
\newcommand{\redto}{\rightsquigarrow} 
\newcommand{\plto}{\xymatrix@C=1pc{\ar[r] & \ar@{-->}@<1ex>[l] }} 
\newcommand{\pllto}{\xymatrix@C=1pc{\ar@{-->}@<-1ex>[r] & \ar[l] }} 
\newcommand{\pleq}{\xymatrix@C=1pc{\ar@{-}[r] & \ar@{--}@<1ex>[l] }} 

\newcommand{\id}{\mathit{id}}
\newcommand{\fr}[2]{{#2}\backslash{#1}} 
\newcommand{\logL}{\mathcal{L}} 
\newcommand{\catC}{\mathbf{C}}
\newcommand{\catS}{\mathbf{S}} 
\newcommand{\catT}{\mathbf{T}} 
\newcommand{\Ss}{\Sigma} 
 
\renewcommand{\ss}{\sigma} 
 
\newcommand{\eqn}{\mathrm{eq}} 
\newcommand{\nat}{\mathrm{nat}}

\newcommand{\potop}{\mathrm{top}} 
\newcommand{\pobot}{\mathrm{bottom}} 
\newcommand{\poback}{\mathrm{back}} 
\newcommand{\pofront}{\mathrm{front}} 
\newcommand{\podiag}{\mathrm{diag}} 
\newcommand{\postep}{\mathrm{step1}} 
\newcommand{\podred}{\mathrm{step2}} 

\title{Deduction as Reduction} 
\author{Dominique Duval \\ 
LJK, University of Grenoble \\ 
\texttt{Dominique.Duval@imag.fr}}
\date{October 30., 2010}

\begin{document}
\maketitle

\begin{itemize}
\item[] \textbf{Abstract.}
Deduction systems and graph rewriting systems
are compared within a common categorical framework. 
This leads to an improved deduction method in diagrammatic logics. 
\end{itemize}

\section{Introduction}

Deduction systems and graph rewriting systems
can be seen as different kinds of reduction systems. 
In this article, they are compared within a common categorical framework. 
This leads to an improved deduction method in diagrammatic logics. 

On the one hand, 
category theory may be used for describing graph rewriting systems:
this is called the algebraic approach to graph rewriting \cite{dpo, spo}.
On the other hand, 
proofs may be modelled by morphisms in a category \cite{lam68, law69}.
However, these are two disjoint research topics. 
In this paper, we present a common categorical framework 
for dealing with graph rewriting and with logical deduction. 

Our deduction systems are defined using \emph{diagrammatic logics} \cite{Du03, DD10},
where a logic is a functor $\logL:\catS\to\catT$ satisfying several properties. 
In this framework, a major role is played by \emph{pleomorphisms}
(called \emph{entailments} in previous papers):
for a given logic $\logL:\catS\to\catT$,  
a pleomorphism is ``half-way'' between a general morphism 
and an isomorphism; it is a (generally non-reversible) morphism in the category
$\catS$ which is mapped by the functor $\logL$ 
to a reversible morphism in the category $\catT$.
In biology, a pleomorphism is the occurrence of several  
structural forms during the life cycle of a plant;
in diagrammatic logic, a pleomorphism refers to the occurrence of several  
presentations of a given logical theory during a proof:
various lemmas are progressively added to the given axioms 
until the required theorem is obtained. 
In this paper, the analogy with graph rewriting systems  
extends this approach in such a way that it becomes possible 
to drop intermediate lemmas. 

Section~\ref{sec:red} is devoted to graph rewriting systems
and section~\ref{sec:ded} to deduction systems,
then the comparison is done in section~\ref{sec:dred}.
Some familiary is assumed from the reader with the notions of 
categories, functors and pushouts.
It is recalled that a span (resp. a cospan) in a category is simply 
a pair of morphisms with the same domain (resp. codomain).

\section{Reduction: graph rewriting}
\label{sec:red}

Well-known examples of \emph{reduction systems} (or \emph{rewriting systems}) 
are \emph{string} rewriting systems, 
\emph{term} rewriting systems and 
\emph{graph} rewriting systems. 
Let us focus on the last ones. 
A graph rewriting system consists of a binary relation $\redto$ on 
graphs, i.e., a set of \emph{rewrite rules} of the form $L\redto R$. 
Given a rewrite rule $L\redto R$ and  
an occurrence (called a \emph{match}) of $L$ in a graph $G$, 
the \emph{rewrite step} consists of ``replacing'' the occurrence of $L$ in $G$ 
by an instance of $R$, which gives rise to a new graph $H$.
This can be applied to various families of graphs, 
and the notion of ``replacement'' may take various meanings.

In the algebraic approach,
graph rewriting systems are described in a categorical framework.
Such systems include the 
\emph{double pushout} (DPO), 
\emph{simple pushout} (SPO), 
\emph{sesqui-pushout} (SqPO) and
\emph{heterogeneous pushout} (HPO) approaches   
\cite{spo, dpo, sqpo, hpo}. 

In this paper we focus on the double pushout and 
the sesqui-pushout graph rewriting systems.
Given a category $\catC$ of graphs, 
in both approaches a \emph{match} is a morphism $m_L:L\rto G$ in $\catC$
and a \emph{rewrite rule} $L\redto R$ is a span $(l,r)$ between $L$ and $R$
in $\catC$
  $$ \xymatrix@R=1pc{
  L \ar[d]_{m_L} \\
  G \\
  } \qquad\qquad
  \xymatrix@R=0pc@C=4pc{
  & K \ar[ld]_{l} \ar[rd]^{r} & \\
  L & & R \\
  } $$
Let us call \emph{generalized pushout under a span $(l,r)$} 
a diagram of the following form, 
with a commutative square on the left and a pushout square on the right
  $$ \xymatrix@R=0pc@C=4pc{
  & K \ar[ld]_{l} \ar[rd]^{r} \ar[dd]^{m_K} & \\
  L \ar[dd]_{m_L} & & R \ar[dd]^{m_R} \\
  & D \ar[ld]_{l_1} \ar[rd]^{r_1} \ar@{}[ul]|{(=)} \ar@{}[ur]|{(PO)} & \\
  G & & H \\
  } $$
In both graph rewriting systems the \emph{rewrite step} 
builds a generalized pushout.
First $m_K$ and $l_1$ are built from $m_L$ and $l$ 
so as to get a commutative square,
then $m_R$ and $r_1$ are built from $m_K$ and $r$ 
so as to get a pushout. 
In both approaches there are restrictions on the form 
of the matches and rules. 
\begin{itemize}
\item In the double pushout approach, the left square 
is a pushout, which means that the construction of 
$m_K$ and $l_1$ from $m_L$ and $l$ is a 
\emph{pushout complement}.
\item In the sesqui-pushout approach, the left square 
is a pullback, and more precisely the construction of 
$m_K$ and $l_1$ from $m_L$ and $l$ is a 
\emph{final pullback complement}.
\end{itemize}

\section{Deduction: diagrammatic logic}
\label{sec:ded}

Deduction systems, in  this paper, 
are defined in the framework of \emph{diagrammatic logic}
\cite{Du03, DD10}. 
We only need to know that a diagrammatic logic
is a pushout-preserving functor $\logL:\catS\to\catT$ 
between two categories with pushouts.
Then $\catT$ is called the category of \emph{theories} of the logic $\logL$
and $\catS$ its category of \emph{specifications}.
Each specification $\Ss$ \emph{presents}, or \emph{generates}, the theory~$\logL\Ss$. 

In order to get a full definition 
one must add that $\logL$ is the left adjoint 
in an adjunction induced by a morphism of limit sketches \cite{Eh68}
and that it makes $\catT$ a category of fractions over $\catS$ \cite{GZ67}. 
This full definition, which will not be used in this paper,
enlightens the importance of pleomorphisms (definition~\ref{defi:ded-pleo})  
in diagrammatic logic: in fact, 
in this situation the pleomorphisms determine the functor $\logL$ \cite{GZ67}. 

\begin{definition}
\label{defi:ded-pleo}
With respect to some given diagrammatic logic $\logL:\catS\to\catT$:
\begin{itemize}
\item Two specifications $\Ss$, $\Ss'$ in $\catS$ are \emph{pleoequivalent} 
if there is an isomorphism of theories $\logL\Ss \cong \logL\Ss'$; 
this is denoted $\Ss\pleq\Ss'$. 
\item An \emph{instance} of $\Ss_1$ in $\Ss_2$, 
where $\Ss_1$ and $\Ss_2$ are specifications, 
is a morphism $\ss':\Ss_1\rto\Ss'_2$ in $\catS$
where $\Ss'_2$ is pleoequivalent to $\Ss_2$;
this is denoted $\Ss_1\rto\Ss'_2\pleq\Ss_2$.
\item A \emph{pleomorphism} 
is a morphism of specifications $\tau:\Ss\rto\Ss'$ 
such that $\logL\tau$ is an isomorphism of theories;
this is denoted $\tau:\Ss\plto\Ss'$.
\item A \emph{fraction} from $\Ss_1$ to $\Ss_2$ 
is a cospan $(\ss:\Ss_1\rto\Ss'_2,\tau:\Ss_2\plto\Ss'_2)$ in $\catS$ 
where $\tau$ is a pleomorphism;
this is denoted $\fr{\ss}{\tau}:\Ss_1 \rto \Ss'_2 \pllto \Ss_2$.
The \emph{numerator} of $\fr{\ss}{\tau}$ is~$\ss$, its \emph{denominator} is $\tau$
and its \emph{vertex} is $\Ss'_2$.
\end{itemize}
\end{definition}

\begin{remark}
Clearly, when two specifications are related by a zig-zag of 
pleomorphisms, they are pleoequivalent:
the equivalence relation generated by the pleomorphisms
is included in the pleoequivalence relation.
\end{remark}

The next result states some straightforward properties of pleomorphisms.

\begin{proposition}
\label{prop:ded-pleo}
Pleomorphisms satisfy the following properties: 
\begin{itemize} 
\item every isomorphism in $\catS$ is a pleomorphism,
\item if $h=g\circ f$ in $\catS$ and if two among $f,g,h$ are 
pleomorphisms, then so is the third,
\item pleomorphisms are stable under pushouts. 
\end{itemize}
\end{proposition}

\begin{definition}
\label{defi:ded-step}
Given a diagrammatic logic $\logL:\catS\to\catT$
\begin{itemize}
\item A \emph{deduction rule} (or \emph{inference rule})
is a fraction $\fr{c}{h}:C \rto P \pllto H$ from $C$ to $H$. 
The \emph{hypothesis} of $\fr{c}{h}$ is $H$, its \emph{conclusion} is $C$.
\item The \emph{deduction step} 
with respect to a deduction rule $\fr{c}{h}:C \rto P \pllto H$ 
maps each instance $\ss_H:H\rto\Ss_H\pleq\Ss$ of $H$ in some specification $\Ss$
to the instance $\ss_C:C\rto\Ss_C\pleq\Ss$ of $C$ in the same $\Ss$
defined as follows (where $h_1$ is a pleomorphism 
because so is $h$ and pleomorphisms are stable under pushouts) 
  $$ 
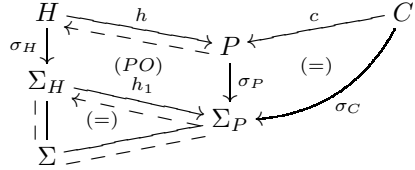
\xymatrix@R=0pc@C=4pc{
  H \ar[rd]^{h} \ar[dd]_{\ss_H} && C \ar@/^3ex/[dddl]^{\ss_C} \ar[ld]_{c} \\
  & P \ar[dd]^{\ss_P} \ar@<1ex>@{-->}[lu] & \\
  \Ss_H \ar[rd]^{h_1} &&  \\
  & \Ss_P \ar@<1ex>@{-->}[lu] \ar@{}[uuul]|{(PO)} \ar@{}[uuur]|{(=)} \ar@{}[l]|(.7){(=)} &  \\
  \Ss \ar@{-}[uu] \ar@<1ex>@{--}[uu] \ar@{-}[ru] \ar@<-1ex>@{--}[ru] && \\
  } $$
\end{itemize}
\end{definition}

\begin{remark}
A deduction rule $\fr{c}{h}$ has numerator $c$ and denominator $h$,
in contrast with the usual notation $\frac{H}{C}$.
Indeed, it is the morphism $h$, and not $c$, which becomes 
an isomorphism of theories. 
\end{remark}

\begin{remark}
This construction is essentially 
the composition of fractions in their bicategory.
\end{remark}

\begin{example}
\label{exam:ded}
Let $\logL_{\eqn}$ be the equational logic.
One of its rule is the \emph{transitivity} rule 
  $$ \frac{x\equiv  y \quad y\equiv z }{x\equiv z}$$
which corresponds to the fraction
  $$ \begin{array}{ccccc}     
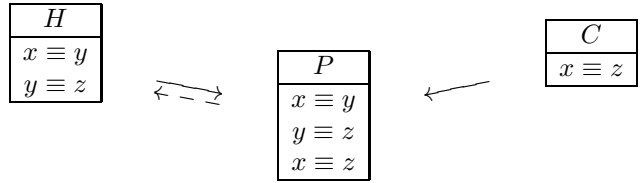
  \begin{array}{|c|} \hline H \\ \hline 
    x\equiv y \\ y\equiv z \\ \hline \multicolumn{1}{c}{\rule{0pt}{20pt}} \\ \end{array} & 
  \begin{array}{c} 
    \xymatrix@R=0pc{ \ar[rd] & \\ & \ar@<1ex>@{-->}[lu] \\ } \end{array} & 
  \begin{array}{|c|} \multicolumn{1}{c}{\rule{0pt}{20pt}} \\ \hline P \\ \hline 
    x\equiv y \\ y\equiv z \\ x\equiv z \\ \hline \end{array} & 
  \begin{array}{c} 
    \xymatrix@R=0pc{ & \ar[ld] \\ & \\} \end{array} & 
  \begin{array}{|c|} \hline C \\ \hline 
    x\equiv z \\ \hline \multicolumn{1}{c}{\rule{0pt}{20pt}}  \\ \end{array} \\ 
  \end{array} $$
Let $\Ss_{\nat}$ be the equational specification ``of naturals'' 
made of a sort $N$, a constant $0:N$, two operations $s:N\to N$, $+:N^2\to N$
and two equations $0+y\equiv y$, $s(x)+y\equiv s(x+y)$.
Let us analyze the last step in the proof of $1+1\equiv s(1)$
(where $1$ stands for $s(0)$), 
once it has been proved that 
$1+1\equiv s(0+1)$ and that $s(0+1)\equiv s(1)$,
so that it remains to use the transitivity rule in order to conclude.
Then $\Ss_{\nat,H}$ is $\Ss_{\nat}$ together with the terms 
$1+1$, $s(0+1)$, $s(1)$ and the equations 
$1+1\equiv s(0+1)$, $s(0+1)\equiv s(1)$.
The deduction step yields $\Ss_{\nat,P}$, made of 
$\Ss_{\nat,H}$ together with the equation $1+1\equiv s(1)$,
and $\ss_C$ maps $x\equiv z$ to $1+1\equiv s(1)$, as required.
  $$ \begin{array}{ccc}     
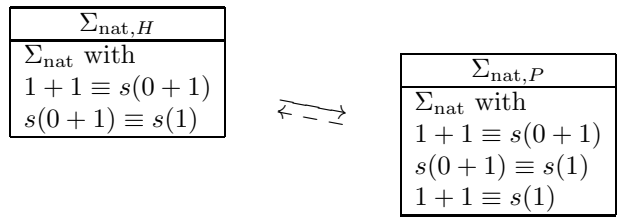
  \begin{array}{|l|} \hline 
    \multicolumn{1}{|c|}{\Ss_{\nat,H}} \\ \hline 
    \Ss_{\nat} \mbox{ with } \\
    1+1\equiv s(0+1) \\ 
    s(0+1)\equiv s(1) \\ 
    \hline \multicolumn{1}{c}{\rule{0pt}{20pt}} \\ \end{array} & 
  \begin{array}{c} 
    \xymatrix@R=0pc{ \ar[rd] & \\ & \ar@<1ex>@{-->}[lu] \\ } \end{array} & 
  \begin{array}{|l|} \multicolumn{1}{c}{\rule{0pt}{20pt}} \\ \hline 
    \multicolumn{1}{|c|}{\Ss_{\nat,P}} \\ \hline 
    \Ss_{\nat} \mbox{ with } \\
    1+1\equiv s(0+1) \\ 
    s(0+1)\equiv s(1) \\ 
    1+1\equiv s(1) \\ \hline \end{array} \\ 
  \end{array} \qquad \qquad $$
\end{example}

\section{Deduction as Reduction}
\label{sec:dred}

It is clear from the previous sections that deduction
is a form of reduction, as well as graph rewriting.
\begin{itemize}
\item In a graph rewriting system (section~\ref{sec:red}),
given a rewrite rule $L\redto R$ and a match of $L$ in $G$, 
the rewrite step consists of ``replacing'' the occurrence of $L$ in a graph $G$ 
by an instance of $R$, which gives rise to a new graph $H$.
\item In a deduction system (section~\ref{sec:ded}),
given a deduction rule $\frac{H}{C}$ and an instance of $H$ in $\Ss$
with vertex~$\Ss_H$, 
the deduction step consists of ``replacing'' the occurrence of $H$ in $\Ss$ 
by an instance of $C$, which gives rise to an instance of $C$ in $\Ss$
with a new vertex $\Ss_C$.
\end{itemize}

But a graph rewrite rule is a span while a deduction rule is a cospan, 
so that the descriptions of the rewrite steps are quite different. 
However, in this section, under the assumption that a deduction rule 
can also be defined from a span, 
we exhibit similarities between both reduction systems
and we propose improvements in the construction of the deduction steps.

\begin{assumption}
\label{ass:dred-po}
It is now assumed 
that each deduction rule $\fr{c}{h}:C \rto P \pllto H$ 
is obtained from a pushout
  $$\xymatrix@R=0pc@C=3pc{
  & K \ar[dl]_{l} \ar[drr]^{r} && \\
  H \ar@<-1ex>[drr]_{h} &&& C \ar[dl]^{c} \\
  && P \ar@{-->}[llu] \ar@{}[uul]|{(PO)} & \\
  }$$ 
\end{assumption}

\begin{remark}
Usually a deduction rule is given as $\frac{H}{C}$, 
neither $P$ nor $K$ are mentioned. 
Then $K$ can be defined as the family of features which have 
the same name in $H$ and in $C$,
and $P$ can be obtained from a pushout as in assumption~\ref{ass:dred-po}.
\end{remark}

In section~\ref{sec:red} we have seen 
double pushouts and sesqui-pushouts as instances of generalized pushouts. 
Now we define a third family of generalized pushouts. 

\begin{definition}
\label{defi:dred-plpo}
Given a span $(l,r)$, a \emph{pleopushout under $(l,r)$} is a diagram of the following form, 
with a commutative square on the left, a pushout square on the right,
where $l_1$ is a pleomorphism
  $$ \xymatrix@R=0pc@C=3pc{
  & K \ar[ld]_{l} \ar[rrd]^{r} \ar[dd]^{\ss_K} && \\
  H \ar[dd]_{\ss_H} & && C \ar[dd]^{\ss_C} \\
  & \Ss_K \ar@<1ex>[ld]^{l_1} \ar[rrd]_{r_1} \ar@{}[ul]|{(=)} \ar@{}[urr]|{(PO)} && \\
  \Ss_H \ar@{-->}[ru] & && \Ss_C \\
  } $$
\end{definition}

Theorem~\ref{theo:dred-plpo} below builds a deduction step from a 
pleopushout. In its proof we use the properties of pleomorphisms 
stated in proposition~\ref{prop:ded-pleo}
and well-known properties of pushouts 
stated now in proposition~\ref{prop:dred-po}.

\begin{proposition}
\label{prop:dred-po}
Given two consecutive commutative squares $(1)$, $(2)$ 
and the composed commutative squares $(3)$,
if $(1)$ is a pushout then 
$(2)$ is a pushout if and only if $(3)$ is a pushout.
  $$ \xymatrix{
  \bullet \ar[r] \ar[d] & \bullet \ar[r] \ar[d] & \bullet \ar[d] \ar@{}[dr]|{=} &
  \bullet \ar[r] \ar[d] & \bullet \ar[r] & \bullet \ar[d] \\
  \bullet \ar[r] & \bullet \ar[r] \ar@{}[ul]|{(1)} \ar@{}[ur]|{(2)} & \bullet &
  \bullet \ar[r] & \bullet \ar[r] \ar@{}[u]|{(3)} & \bullet \\
  }$$
\end{proposition}

\begin{theorem}
\label{theo:dred-plpo}
Let $\fr{c}{h}$ be a deduction rule satisfying assumption~\ref{ass:dred-po}
  $$\xymatrix@R=0pc@C=3pc{
  & K \ar[dl]_{l} \ar[drr]^{r} && \\
  H \ar@<-1ex>[drr]_{h} &&& C \ar[dl]^{c} \\
  && P \ar@{-->}[llu] \ar@{}[uul]|{(PO)_\potop} & \\
  }$$ 
 Let $\ss_H$ be an instance of $H$ in $\Ss$,
 and let us assume that there is a pleopushout under $(l,r)$ 
  $$ \xymatrix@R=0pc@C=3pc{
  & K \ar[ld]_{l} \ar[rrd]^{r} \ar[dd]^{\ss_K} && \\
  H \ar[dd]_{\ss_H} & && C \ar[dd]^{\ss_C} \\
  & \Ss_K \ar@<1ex>[ld]^{l_1} \ar[rrd]_{r_1} \ar@{}[ul]|{(=)} \ar@{}[urr]|{(PO)_\poback} && \\
  \Ss_H \ar@{-->}[ru] & && \Ss_C \\
  } $$
Let us consider the pushout
  $$\xymatrix@R=0pc@C=3pc{
  & \Ss_K \ar@<-1ex>[dl]_{l_1} \ar[drr]^{r_1} && \\
  \Ss_H \ar@{-->}[ru] \ar[drr]_{h_1} &&& \Ss_C \ar[dl]^{c_1} \\
  && \Ss_P \ar@{}[uul]|{(PO)_\pobot} & \\
  }$$ 
Then there is a unique morphism $\ss_P:P\rto\Ss_P$ such that we get a commutative cube
  $$ \xymatrix@R=0pc@C=3pc{
  & K \ar[ld]_{l} \ar[rrd]^(.3){r} \ar@{..>}[dddd]^(.2){\ss_K} && \\
  H \ar[dddd]_(.2){\ss_H} \ar[drr]_(.3){h} &&& 
    C \ar[dddd]^(.2){\ss_C} \ar[dl]^{c} \\
  && P \ar[dddd]_(.2){\ss_P} & \\
  &&& \\ 
  & \Ss_K \ar@{..>}[ld]_{l_1} \ar@{..>}[rrd]^(.3){r_1} && \\
  \Ss_H \ar[drr]_(.3){h_1} &&& 
    \Ss_C \ar[dl]^{c_1} \\
  && \Ss_P & \\
  } $$
Then in this cube:
\begin{itemize}
\item the top, bottom, back right and front left faces are pushouts, 
\item the morphism $h$ and the four morphisms of the bottom face are pleomorphisms.
\end{itemize}
If $\ss_H$ is an instance of $H$ in some $\Ss$ 
then $\ss_C$ is an instance of $C$ in $\Ss$.
\end{theorem}

\begin{proof}
It is easily checked that the following square is commutative
  $$\xymatrix@R=0pc@C=3pc{
  & K \ar[dl]_{l} \ar[drr]^{r} && \\
  H \ar[drr]_{h_1\circ\ss_H} &&& C \ar[dl]^{c_1\circ\ss_C} \\
  && \Ss_P \ar@{}[uul]|{(=)} & \\
  }$$ 
So, the pushout $(PO)_\potop$ gives rise to a unique morphism 
$\ss_P:P\rto\Ss_P$ such that the two front faces in the cube are commutative.
Since the other faces are yet known to be commutative,
the cube is commutative.

It is yet known that the top $(PO)_\potop$, bottom $(PO)_\pobot$
and back right $(PO)_\poback$ faces are pushouts.
Let us prove that the front left face is also a pushout.  
According to proposition~\ref{prop:dred-po}, 
composing $(PO)_\poback$ and $(PO)_\pobot$ gives rise to the 
``diagonal'' pushout $(PO)_\podiag$
  $$ \xymatrix@R=0pc@C=3pc{
  & K \ar[ld]_{l_1\circ\ss_K} \ar[rrd]^{r} && \\
  \Ss_H \ar[drr]_{h_1} &&& 
    C \ar[dl]^{c_1\circ\ss_C} \\
  && \Ss_P \ar@{}[uul]|{(PO)_\podiag} & \\
  } $$
Thanks to the commutativity of the cube, the pushout $(PO)_\podiag$ can also be written as
  $$ \xymatrix@R=0pc@C=3pc{
  & K \ar[ld]_{\ss_H\circ l} \ar[rrd]^{r} && \\
  \Ss_H \ar[drr]_{h_1} &&& 
    C \ar[dl]^{\ss_P\circ c} \\
  && \Ss_P \ar@{}[uul]|{(PO)_\podiag} & \\
  } $$
Now, since $(PO)_\potop$ and $(PO)_\podiag$ are pushouts,
proposition~\ref{prop:dred-po} implies that the 
front left face of the cube is also a pushout $(PO)_\pofront$.

It is yet known that $h$ and $l_1$ are pleomorphisms. 
Let us check that the three other morphisms 
of the bottom face are pleomorphisms,
using proposition~\ref{prop:ded-pleo}.
Since pleomorphisms are stable under pushouts, 
it follows from $(PO)_\pobot$ that $c_1$ is a pleomorphism
and from $(PO)_\pofront$ that $h_1$ is a pleomorphism.
Since three among the four morphisms in the bottom commutative square 
are pleomorphisms, so is the fourth: 
hence $r_1$ is also a pleomorphism.

It follows that $\Ss_H$ and $\Ss_C$ are pleoequivalent,
which proves the last assertion of the theorem.
\end{proof}

\begin{remark}
According to definition~\ref{defi:ded-step}, 
the deduction step with respect to $\fr{c}{h}$ is defined from the following diagram 
  $$ \xymatrix@R=0pc@C=4pc{
  H \ar[rd]^{h} \ar[dd]_{\ss_H} && C \ar@/^3ex/[dddl]^{\ss_C} \ar[ld]_{c} \\
  & P \ar[dd]^{\ss_P} \ar@<1ex>@{-->}[lu] & \\
  \Ss_H \ar[rd]^{h_1} &&  \\
  & \Ss_P \ar@<1ex>@{-->}[lu] \ar@{}[uuul]|{(PO)_\postep} \ar@{}[uuur]|{(=)} &  \\
  }$$
If $\fr{c}{h}$ satisfies assumption~\ref{ass:dred-po},
then by proposition~\ref{prop:dred-po} the composition of 
$(PO)_\potop$ and $(PO)_\postep$ yields a pushout $(PO)_\podred$,
which obviously forms the right part of a pleopushout under $(l,r)$
  $$ \xymatrix@R=0pc@C=3pc{
  & K \ar[dl]_{l} \ar[drr]^{r} \ar[dd]^{\ss_H\circ l} && \\
  H \ar[dd]_{\ss_H} &&& C \ar[dd]^{\ss_C} \\
  & \Ss_H \ar@<1ex>[dl]^{\id} \ar[drr]^{h_1} \ar@{}[ul]|{(=)} \ar@{}[urr]|{(PO)_\podred} && \\
  \Ss_H \ar@{-->}[ru] &&& \Ss_P \ar@<1ex>@{-->}[ull] \\
  } $$
So, a pleopushout under $(l,r)$ as assumed in theorem~\ref{theo:dred-plpo}
is obtained from the deduction step (definition~\ref{defi:ded-step}).
This proves that indeed a deduction step can be seen as a reduction step.
Moreover, theorem~\ref{theo:dred-plpo} 
states that whenever we are able to find a ``better'' pleopushout
than this obvious one, 
then we may get an instance of $C$ in $\Ss$ ``better'' than $\ss_C:C\rto\Ss_P$
in definition~\ref{defi:ded-step}.
Such a situation occurs in example~\ref{exam:dred}.
\end{remark}

\begin{example}
\label{exam:dred}
As in example~\ref{exam:ded}, 
let $\logL_{\eqn}$ be the equational logic.
The \emph{transitivity} rule:
  $$ \frac{x\equiv  y \quad y\equiv z }{x\equiv z}$$
can be obtained by a pushout from a span $H \lto K \rto C$:
  $$ \begin{array}{ccccc}
  \begin{array}{|c|} \multicolumn{1}{c}{\rule{0pt}{20pt}} \\ \hline H \\ \hline 
    x\equiv y \\ y\equiv z \\ \hline \end{array} & 
  \begin{array}{c} 
    \xymatrix@R=0pc{ & \ar[ld] \\ & \\} \end{array} & 
  \begin{array}{|c|} \hline K \\ \hline 
    x \\ z \\ \hline \multicolumn{1}{c}{\rule{0pt}{20pt}}  \\ \end{array} &
  \begin{array}{c} 
    \xymatrix@R=0pc{ \ar[rd] & \\ & \\ } \end{array} & 
  \begin{array}{|c|} \multicolumn{1}{c}{\rule{0pt}{20pt}}  \\ \hline C \\ \hline 
    x\equiv z \\ \hline \end{array} \\ 
  \end{array} $$
As in example~\ref{exam:ded}, 
let $\Ss_{\nat}$ be the equational specification ``of naturals''
and let us analyze the last step in the proof of $1+1\equiv s(1)$:
it has yet been proved that 
$1+1\equiv s(0+1)$ and $s(0+1)\equiv s(1)$, 
and it remains to use the transitivity rule. 
As in example~\ref{exam:ded}, 
$\Ss_{\nat,H}$ is $\Ss_{\nat}$ together with the terms 
$1+1$, $s(0+1)$, $s(1)$ and the equations 
$1+1\equiv s(0+1)$, $s(0+1)\equiv s(1)$.
Let us define $\Ss_{\nat,K}$ as 
$\Ss_{\nat}$ with the terms $1+1$ and $s(1)$
(with no additional equations),
with $\ss_{\nat,K}$ which maps $x$ to $1+1$ and $z$ to $s(1)$
and with $l_1$ the inclusion.
Then by pushout $\Ss_{\nat,C}$ is made of $\Ss_{\nat}$ 
with the equation $1+1\equiv s(1)$. 
It is smaller than $\Ss_{\nat,P}$ from example~\ref{exam:ded}:
the lemmas $1+1\equiv s(0+1)$ and $s(0+1)\equiv s(1)$,
which have been used during the proof, 
are kept in $\Ss_{\nat,P}$ while they are dropped from $\Ss_{\nat,C}$. 
  $$ \begin{array}{ccccc}     
  \begin{array}{|l|} \hline 
    \multicolumn{1}{|c|}{\Ss_{\nat,H}} \\ \hline 
    \Ss_{\nat} \mbox{ with } \\
    1+1\equiv s(0+1) \\ 
    s(0+1)\equiv s(1) \\ 
    \hline \multicolumn{1}{c}{\rule{0pt}{20pt}} \\ \end{array} & 
  \begin{array}{c} 
    \xymatrix@R=0pc@C=1pc{ \ar[rd] & \\ & \ar@<1ex>@{-->}[lu] \\ } \end{array} & 
  \begin{array}{|l|} \multicolumn{1}{c}{\rule{0pt}{20pt}} \\ \hline 
    \multicolumn{1}{|c|}{\Ss_{\nat,P}} \\ \hline 
    \Ss_{\nat} \mbox{ with } \\
    1+1\equiv s(0+1) \\ 
    s(0+1)\equiv s(1) \\ 
    1+1\equiv s(1) \\ \hline \end{array} &
  \begin{array}{c} 
    \xymatrix@R=0pc@C=1pc{ & \ar[ld]  \\ \ar@<-1ex>@{-->}[ru] & \\  } \end{array} & 
  \begin{array}{|l|} \hline 
    \multicolumn{1}{|c|}{\Ss_{\nat,C}} \\ \hline 
    \Ss_{\nat} \mbox{ with } \\
    1+1\equiv s(1) \\ \hline 
    \multicolumn{1}{c}{\rule{0pt}{20pt}} \\ \end{array} \\
  \end{array} $$
\end{example}

\section{Conclusion}
\label{sec:conc}

Deduction systems as well as graph rewriting systems
can be seen as reduction systems. 
This paper lays the foudations for such comparisons. 
Further developments might involve adhesive categories \cite{adhcat}.   
This should provide a new point of view about the role of
pullbacks in graph rewriting, as well as new methods for 
deduction in diagrammatic logics.


\end{document}